\documentclass[preprint]{vgtc}               

\ifpdf
  \pdfoutput=1\relax                   
  \usepackage{graphicx}                
  \DeclareGraphicsExtensions{.pdf,.png,.jpg,.jpeg} 
\else
  \usepackage{graphicx}                
  \DeclareGraphicsExtensions{.eps}     
\fi%

\graphicspath{{figures/}{images/}{./}}

\usepackage{microtype}                 
\PassOptionsToPackage{warn}{textcomp}  
\usepackage{textcomp}                  
\usepackage{mathptmx}                  
\usepackage{times}                     
\usepackage{cite}                      
\usepackage{tabu}                      
\usepackage{booktabs}                  

\usepackage{graphicx}
\usepackage{relsize}
\usepackage[percent]{overpic}
\usepackage{mathptmx}
\usepackage{xcolor}
\usepackage{amsmath,amsfonts,amssymb}

\usepackage{caption}
\usepackage{subcaption}

\usepackage{tikz, tikz-3dplot}
\usetikzlibrary{calc}
\usepackage{pdfpages}
\usepackage{array,booktabs,makecell}

\usepackage[ruled,linesnumbered,vlined,boxed,commentsnumbered]{algorithm2e}
\usepackage{setspace}

\usepackage{xcolor}
\usepackage{verbatim}
\usetikzlibrary{positioning}
\usepackage{hyperref}
\usepackage{lipsum}
\usepackage{blindtext}

\usepackage{amsthm}

\newtheorem{theorem}{Theorem}

\renewcommand{\vec}[1]{\mathop{}\overrightarrow{#1}}

\newtheorem*{EDL}{Even Degree Lemma}{\bfseries}{}

\newcommand{\sig}       {\mm{\textrm{sgn}}}

\newcommand {\mm}[1] {\ifmmode{#1}\else{\mbox{\(#1\)}}\fi}

\newcommand{\Mspace}        {\mm{{\mathbb M}}}
\newcommand{\Jspace}        {\mm{{\mathbb J}}}
\newcommand{\Rspace}        {\mm{{\mathbb R}}}

\renewcommand{\star}        {\mm{\textrm{St}}}

\newcommand{\link}       {\mm{\textrm{Lk}}}

\newcommand{\x}{\mm{\mathbf{x}}}
\renewcommand{\a}{\mm{\mathbf{a}}}
\renewcommand{\b}{\mm{\mathbf{b}}}
\renewcommand{\v}{\mm{\mathbf{v}_1}}
\newcommand{\vv}{\mm{\mathbf{v}_2}}
\newcommand{\m}{\mm{\mathbf{m}}}
\newcommand{\abv}{\mm{\mathbf{abv}_1}}
\newcommand{\avvb}{\mm{\mathbf{av}_2\mathbf{b}}}
\renewcommand{\vec}[1]{\mm{\mathbf{#1}}}
\newcommand{\ab}{{\a\b}}
\newcommand{\bil}{{\mathrm{bi}}}
\newcommand{\lin}{{\mathrm{li}}}

\definecolor{green}{rgb}{0.0, 0.6, 0.0}
\definecolor{purple}{rgb}{0.489338391672278,   0.385854299896855,   0.484000480546439}
\definecolor{blue}{rgb}{0.230032653793089,   0.298999333865925,   0.754001759886747}
\definecolor{red}{rgb}{0.705997696240579,   0.016126465515202,  0.150001123391885}

\definecolor{blueish}{rgb}{0.615016326896545,   0.649499666932962,   0.877000879943374}
\definecolor{redish}{rgb}{0.852998848120289,   0.508063232757601,   0.575000561695942}

\newcommand\thefont{\expandafter\string\the\font}

\onlineid{1002}

\vgtccategory{Research}

\vgtcinsertpkg

\ieeedoi{10.1109/TopoInVis57755.2022.00011}

\title{Reduced Connectivity for Local Bilinear Jacobi Sets}

\author{Daniel Klötzl\thanks{e-mail: daniel.kloetzl@visus.uni-stuttgart.de}\\ %
     \scriptsize University of Stuttgart (VISUS) %
     \and Tim Krake\thanks{e-mail: tim.krake@visus.uni-stuttgart.de}\\ %
     \scriptsize University of Stuttgart (VISUS) %
     \and Youjia Zhou\thanks{e-mail: zhou325@sci.utah.edu}\\
     \scriptsize University of Utah (SCI) %
     \and Jonathan Stober\thanks{e-mail: jonathan.stober@itlr.uni-stuttgart.de}\\
     \scriptsize University of Stuttgart (ITLR)
     \and Kathrin Schulte\thanks{e-mail: kathrin.schulte@itlr.uni-stuttgart.de}\\
     \scriptsize University of Stuttgart (ITLR)%
     \and Ingrid Hotz\thanks{e-mail: ingrid.hotz@liu.se}\\ %
     \scriptsize Linköping University (ITN) %
     \and Bei Wang\thanks{e-mail: beiwang@sci.utah.edu}\\ %
     \scriptsize University of Utah (SCI) %
     \and Daniel Weiskopf\thanks{e-mail: daniel.weiskopf@visus.uni-stuttgart.de}\\ %
     \scriptsize University of Stuttgart (VISUS) }

\teaser{
    \centering
	\includegraphics[width=\textwidth]{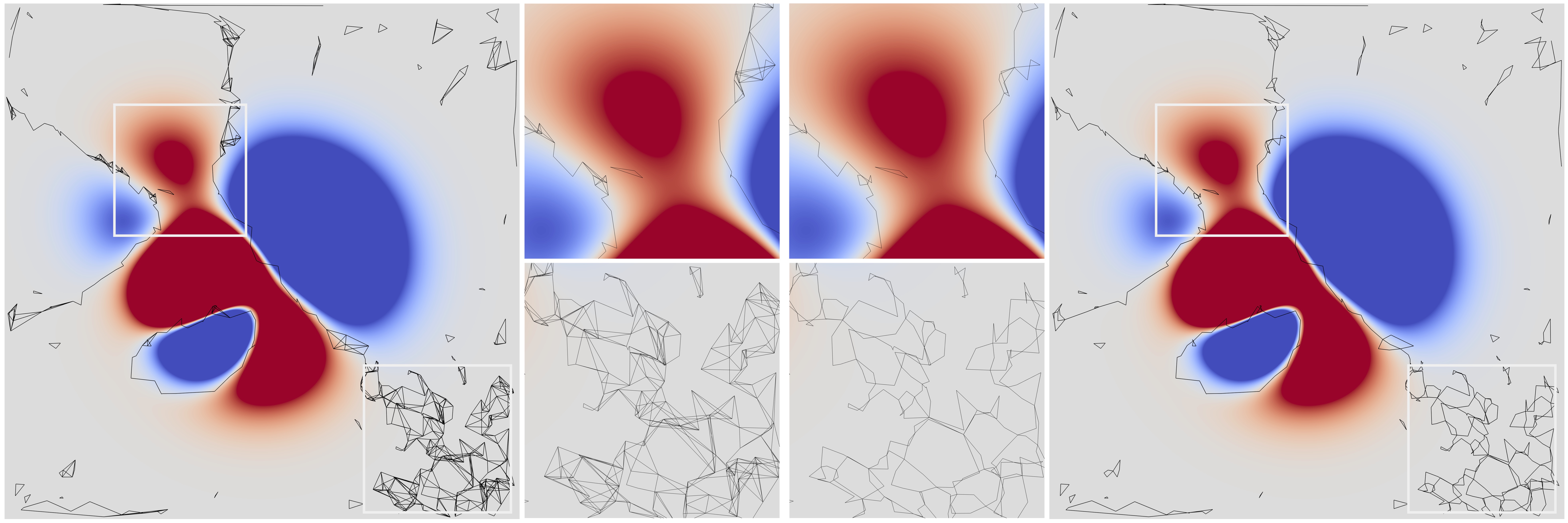}
	\caption{Comparison of the original connectivity construction for a local bilinear Jacobi set (left) and our novel reduced connectivity (right). The color coding (blue-to-red) shows the gradient alignment field of two analytic functions, and the extracted/computed Jacobi set is displayed by solid black lines. The center images show zoomed-in details of two regions in the visualizations (marked by white outlines). It can be observed that the reduced connectivity construction results in fewer edges and, therefore, in a clearer visual representation.}
	\label{fig:teaser}
}

\abstract{
We present a new topological connection method for the local bilinear computation of Jacobi sets that improves the visual representation while preserving the topological structure and geometric configuration. To this end, the topological structure of the local bilinear method is utilized, which is given by the nerve complex of the traditional piecewise linear method. Since the nerve complex consists of higher-dimensional simplices, the local bilinear method (visually represented by the 1-skeleton of the nerve complex) leads to clutter via crossings of line segments. Therefore, we propose a homotopy-equivalent representation that uses different collapses and edge contractions to remove such artifacts. Our new connectivity method is easy to implement, comes with only little overhead, and results in a less cluttered representation.
} 

\CCScatlist{
  \CCScatTwelve{Human-centered computing}{Visu\-al\-iza\-tion}{Visu\-al\-iza\-tion techniques}{};
  \CCScatTwelve{Mathematics of computing}{Dis\-crete mathe\-ma\-tics}{}{}
}

\begin{document}

\newpage

\maketitle

\section{Introduction}

The topological study of multiple scalar fields comprises a broad range of mathematical instruments.
One of these tools is the Jacobi set, a topological descriptor that is based on Morse theory.
For two Morse functions, it is defined as the set of points where the gradients align. So far, there are two computational methods for the numerical computation and extraction of Jacobi sets. 
Edelsbrunner and Harer~\cite{EdelsbrunnerHarer2002} introduced an edge-based computation of Jacobi sets that is based on a \emph{piecewise linear (PL)} scheme.
This PL method produces a non-smooth and inaccurate representation of the Jacobi set due to an edge-based identification and PL interpolation.
As a result, this leads to zig-zag patterns and discretization artifacts. 

Based on the PL method, Klötzl et al.~\cite{Kloetzl:2022:LocalBilinearJS} presented a \emph{local bilinear (LB)} computation of Jacobi sets that enhances the geometry while maintaining the topology of the PL method.
The LB method introduces so-called Jacobi set points, a geometrically more precise (in terms of the gradient alignment measure~\cite{Kloetzl:2022:LocalBilinearJS}) representation of the Jacobi set.
The connection of these points leads to a line segment representation of the Jacobi set that is equal to the nerve complex of the PL method.
As the nerve complex consists of higher-dimensional simplices, the visual representation with its 1-skeleton description has clutter via crossings of line segments (see Figure~\ref{fig:teaser} (left), in particular, bottom right corner).

In this paper, we want to address these issues to obtain an improved visual representation while preserving the topological structure and geometrical configuration.
To achieve this, we apply the theory of topological collapses to the LB method, leading to different collapsing strategies for the higher-dimensional simplices.
Our proposed solution is a combination of these, resulting in a new homotopy-equivalent representation of the LB method with \emph{reduced connectivity} (see Figure~\ref{fig:teaser} (right)).
This new connectivity construction method preserves the topology and inherits important properties, such as the \hyperref[lem:EDL]{Even Degree Lemma}.
In addition, the geometrical configuration is maintained and the representation is less cluttered due to its reduced connectivity.
The proposed algorithm to compute the connectivity comes with only little overhead and is easy to implement.

The remainder of the paper is structured as follows: 
After a discussion of related work in Section~\ref{sec:related-work}, Section~\ref{sec:background} recaps topological foundations (simplicial collapse, strong collapse, and edge contraction) and summarizes the original LB method for the computation of Jacobi sets.
In Section~\ref{sec:our-approach}, we propose our new connectivity construction method, starting with design goals that clarify the desired topological and geometrical properties.
Then, our method is derived with the help of topological collapses and edge contractions.
We conclude this section with the formulation of the algorithm.
In the subsequent section, our method is evaluated qualitatively and quantitatively.
Finally, Section~\ref{sec:conclusion} concludes the paper.

\section{Related Work}
\label{sec:related-work}

The Jacobi set~\cite{EdelsbrunnerHarer2002} is an important tool in topological data analysis that describes the relationship between multiple scalar fields in terms of their gradient behavior. 
It has been employed in various scientific applications to model the relationship between salinity and temperature of water in oceanography~\cite{BarnettPierceSchnur2001}, define the critical paths of gravitational potentials of celestial bodies~\cite{EdelsbrunnerHarer2002}, detect tree rings~\cite{MakelaOpheldersQuigley2020}, extract ridges in images~\cite{NorgardBremer2012}, track features in time-varying simulations~\cite{BremerBringaDuchaineau2007}, and estimate the interrelationships between geophysical multi-fields~\cite{Aramonova2017}, to name a few. 

The classic algorithm that computes the Jacobi set~\cite{EdelsbrunnerHarer2002} is designed for piecewise linear functions defined on triangulations, which is known to produce a large number of discretization artifacts (such as small loops and zig-zag patterns) that can skew the analysis~\cite{BhatiaWangNorgard2015}.   
Recent work by Kl{\"o}tzl et al.~\cite{Kloetzl:2022:LocalBilinearJS} improves upon the piecewise linear approach by utilizing bilinear interpolation to obtain a smoother and more accurate geometric representation of the Jacobi set while preserving its topology. 
Our paper adopts their bilinear computation set but extends their approach with a new topological connection method that improves its visual representation while preserving its topological structure. 

Instead of improving upon interpolation schemes to enhance the Jacobi set representation, a number of previous works introduce controlled simplification of a Jacobi set, oftentimes by ranking and removing parts of the Jacobi sets. 
The \emph{indirect simplification} of the Jacobi set (e.g.,~\cite{BremerBringaDuchaineau2007,LuoSafWan2009}) simplifies the underlying functions to obtain a structurally and geometrically simpler representation, whereas the  \emph{direct simplification} (e.g.~\cite{BhatiaWangNorgard2015, EdelsbrunnerHarerNatarajan2004, NagarajNatarajan2011}) aims to identify and remove  portions of the Jacobi set that are deemed unimportant. 
Our work is different from Jacobi set simplification approaches as it employs simplicial collapses to improve the bilinear representation while preserving the geometry that is lost under simplification schemes.

Finally, a \emph{simplicial collapse} first introduced by Whitehead~\cite{Whitehead:1938tp} is a topological method that reduces a simplicial complex to a subcomplex that is homotopy-equivalent. It has been applied in, e.g.,~computational homology~\cite{Kaczynski:2004we},  discrete Morse theory~\cite{Forman1998}, and ~topological data analysis, e.g., in the study of homotopy types~\cite{AdamaszekAdamsGasparovic2020}.

\section{Background}
\label{sec:background}

In this section, we first recap the theory of simplicial complexes and collapses. 
For this, we use the manuscripts by Forman~\cite{forman2002user} and Boissonnat et al.~\cite{boissonnat_et_al:LIPIcs:2018:9530}.
For a detailed introduction to these topics, we refer to the publications by Whitehead~\cite{Whitehead:1938tp}, Milnor~\cite{milnor1966whitehead}, and Cohens~\cite{cohen2012course}.
Besides the topological aspects, we provide a summary of the LB Jacobi set computation by Klötzl et al.~\cite{Kloetzl:2022:LocalBilinearJS}, which is the mathematical and algorithmic foundation of this paper.

\subsection{Simplicial Complex and Collapse}

To derive our method, we need the theory of simplicial collapse, strong collapse, and edge contraction.
Before we explain these topics, we recap necessary topological definitions such as simplicial complex, simplicial cone, and the nerve of a simplicial complex.
For a more detailed introduction, we refer to the book by Hatcher~\cite{Hatcher2002}.

\paragraph{Simplicial complex}

Given a non-empty finite set $X$, a \emph{simplicial complex} $K$ is given by a collection of subsets of $X$ such that for every subset $L\subset K$, all the subsets of $L$ are contained in $K$. An element $\sigma \in K$ with the cardinality of $k+1$ is denoted as \emph{$k$-simplex}. 
If $\sigma \subseteq \sigma'$, $\sigma$ is called a \emph{face} of $\sigma'$ and $\sigma'$ a \emph{coface} of $\sigma$.
The notion of a \emph{maximal} simplex is given, if it is not a face of any other simplex in $K$.
A sub-collection $L$ of $K$ is a simplicial complex called \emph{subcomplex}.
\emph{Simplicial maps} are defined as maps between two simplicial complexes $\phi \colon K\to L$ and induced by vertex-to-vertex maps $h\colon V(K) \to V(L)$ if the images of the vertices of every simplex in $K$ span a simplex in $L$.

For a simplex $\sigma$ contained in a simplicial complex $K$, the \emph{closed star} $\star_K(\sigma)$ is defined as a subcomplex of $K$ with $\star_K(\sigma):=\{ \tau \in K \mid \tau \cup \sigma \in K \}$ and the link $\link_K(\sigma):=\{ \tau \in \star_K(\sigma) \mid \tau \cap \sigma = \emptyset \}$ is the set of simplices in $\star_K(\sigma)$ that do not intersect with $\sigma$.
The \emph{join} of simplicial complexes is defined as the disjoint union of the two spaces, where every point of one space is attached by line segments to every point of the other space.
In this context, the join of a vertex with a simplicial complex is called \emph{simplicial cone}. For a given simplicial complex $L$ and a vertex $a\notin L$ the simplicial cone $aL$ is defined as $aL:=\{  a, \tau \mid \tau \in L \text{ or } \tau = \sigma \cup a \text{ for } \sigma \in L \}$, where $a,\tau$ represents a simplex.

An important construction is the \emph{nerve} of a simplicial complex $K$, $\mathcal{N}(K)$.
For a simplicial complex $K$, the vertices of $\mathcal{N}(K)$ are given by the maximal simplices of $K$ and the simplices of $\mathcal{N}(K)$ by their non-empty intersection. 
The nerve can be defined iteratively for $j\ge 2$ as $\mathcal{N}^j(K)=\mathcal{N}(\mathcal{N}^{j-1}(K))$ with $\mathcal{N}^1(K):= \mathcal{N}(K)$.
An important property of nerves of simplicial complexes is their connection to the strong collapse, which will be explained later on.

\paragraph{Simplicial collapse}
\begin{figure}[!t]
    \centering
    \includegraphics[width=\columnwidth]{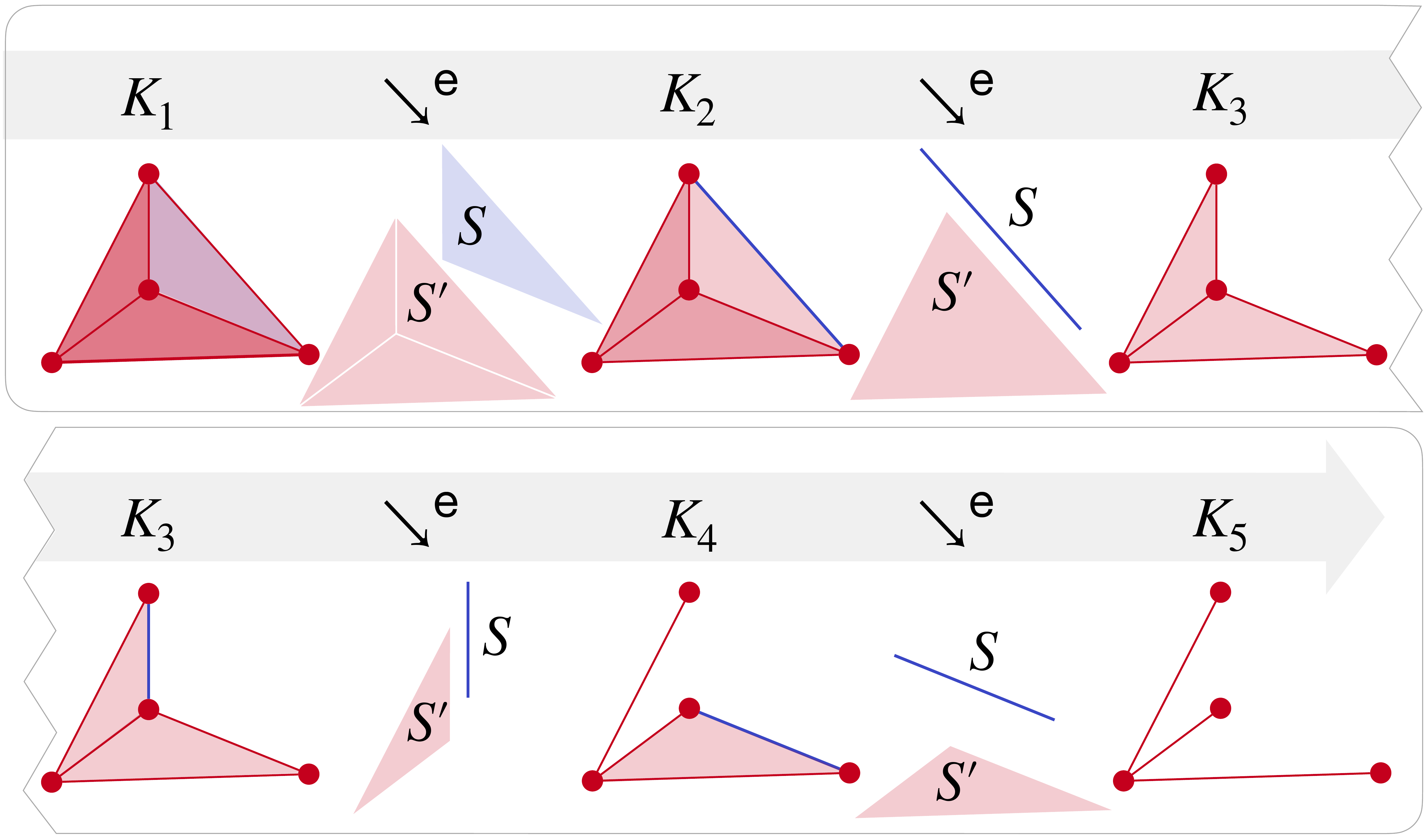}
    \caption{Simplicial collapse of a tetrahedral simplicial 3-complex to a graph (1-complex). 
    The procedure uses four elementary collapses.}
	\label{fig:el_coll}
\end{figure}

Let $K$ be a simplicial complex and $L$ a subcomplex $L\subset K$ of $K$. 
The simplicial complex $K$ collapses to $L$ with an elementary simplicial collapse if there are only two simplices $S$, $S' \in K$ with $S\cap L = \emptyset$ and $S'\cap L = \emptyset$, where $S$ is a \emph{free face} of $S'$, i.e., $S'$ is the unique simplex of $K$ with $S$ as a face.
The elementary simplicial collapse of $K$ to $L$ is denoted as $K \searrow^{\text{e}} L$. 
In general, a simplicial complex $K$ simplicially collapses to $L$, $K \searrow L$, or $L$ \emph{expands} to $K$ if there exists a sequence of finite subcomplexes $K_1, \ldots, K_n$ with $K=K_1$ and $K_n = L$ such that $K_i \searrow^{\text{e}} K_{i+1}$ for all $i\in \{1, \ldots, n-1\}$. 
It can be proven that $K$ collapses to $L$ if and only if the two simplicial complexes are simple homotopy-equivalent~\cite{whitehead1950simple}.

To illustrate the simplicial collapse, Figure~\ref{fig:el_coll} shows a tetrahedral simplicial 3-complex that is simplicially collapsed via four elementary collapses to a graph (1-complex). 
The elementary collapses are always characterized by two simplices that are ``deleted'' from the complex K. In the first elementary collapse, the respective simplices $S'$ and $S$ are the 3-cell and one of the 2-cell faces.

\paragraph{Strong collapse}
\begin{figure}[!t]
    \centering
    \includegraphics[width=\columnwidth]{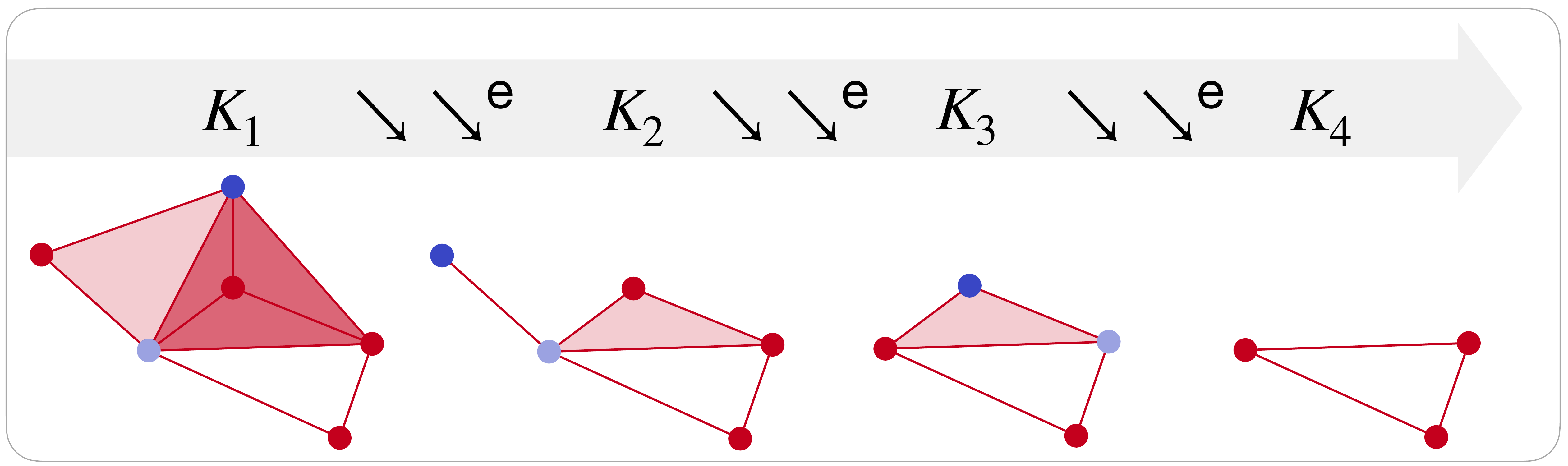}
    \caption{Strong collapse of a simplicial complex. The dark blue vertices are dominated by the light blue vertices, respectively.}
	\label{fig:strong_collapse}
\end{figure}

For the definition of a strong collapse, we need the notion of dominated vertices. 
Given a simplicial complex $K$, a vertex $\vec{a}\in K$ is a \emph{dominated vertex} if the link $\link_K(\vec{a})$ is a simplicial cone. 
This means that a vertex $\vec{a}$ is dominated if there exists a vertex $\vec{a}'$, $\vec{a}\neq \vec{a}'$, and a subcomplex $L\in K$ such that $\link_K(\vec{a})=\vec{a}'L$.
In this case, the vertex $\vec{a}$ is dominated by the vertex $\vec{a}'$. 
Another equivalent formulation is given by the notion of maximal simplices. If and only if all the maximal simplices of $K$ containing $\vec{a}$ also contain $\vec{a}'$, the vertex $\vec{a}\in K$ is dominated by $\vec{a}'\in K$~\cite{barmak2012strong}.

Analogously to the simplicial collapse, the strong collapse is again defined via elementary strong collapses. 
To perform an \emph{elementary strong collapse} one identifies and deletes a dominated vertex $\vec{a}$ from $K$: $K\searrow \searrow^{\text{e}} K \setminus \vec{a}$. 
An example of a strong collapse, derived from three elementary collapses, is given in Figure~\ref{fig:strong_collapse}. 
The \emph{strong collapse} of a simplicial complex $K$ to a subcomplex $L$, which we denote as $K \searrow \searrow L$, is thus defined as a series of elementary strong collapses.

The strong collapse of $K$ to $L$ leads to the same strong homotopy type of $K$ and $L$, and it is also well-known that if $K$ and $L$ have the same strong homotopy type, this implies the same simple homotopy type (but not vice versa).
Another important aspect is that for a simplicial complex $K$, there exists a subcomplex $L$ isomorphic to $\mathcal{N}^2(K)$, such that $K\searrow\searrow L$~\cite[Proposition 3.4]{barmak2012strong}.

\paragraph{Edge contraction}
\begin{figure}[!t]
    \centering
    \includegraphics[width=\columnwidth]{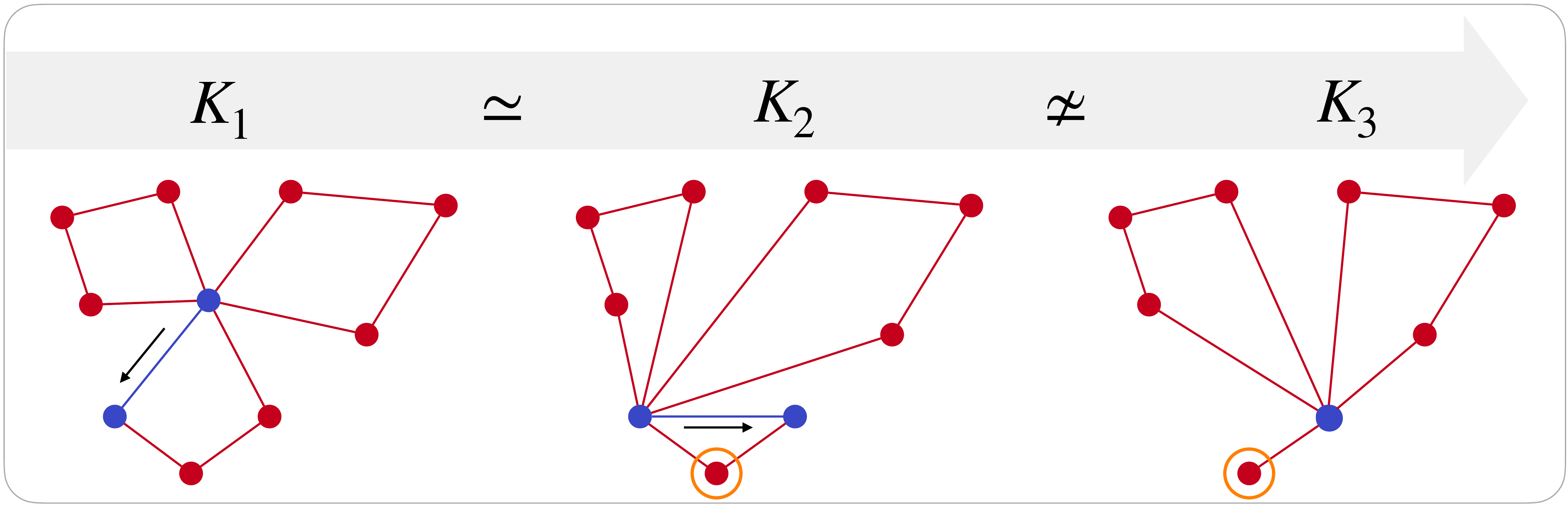}
    \caption{Edge contractions of different simplicial complexes, where contracted edges are marked in blue. 
    The left and middle complexes are homotopy-equivalent, while the middle and right complexes are not. The shared element of the two links contradicting the link condition is encircled in orange.}
	\label{fig:edge_contraction}
\end{figure}

For an edge $\vec{ab}$ of a simplicial complex $K$, an edge contraction is a simplicial map $\phi\colon K \to L$ induced by the vertex map $h_{\vec{ab}}\colon V(K)\to V(L)$ mapping $h(\vec{b}) = \vec{a}$ and everywhere else to identity.
The topology change for edge contractions is investigated by Dey et al.~\cite{dey1998topology} for different configurations.
One of the most important results for our context will be the following \emph{link condition} as a sufficient condition for the preservation of the topology for the edge contraction of certain simplicial complexes.
Given a 1-complex $K$, if $\link_K(\vec{a}) \cap \link_K(\vec{b})=\emptyset$, the edge contraction of $\vec{ab}$ from $K$ to $L$ is a homotopy equivalence.
This result is applied in the illustrative example given in Figure~\ref{fig:edge_contraction}.
In the first contraction ($K_1 \to K_2$), the link of the two vertices (marked in blue) is disjoint.
This leads to a homotopy-equivalent contraction of the edge to a new vertex, whereas the second contraction ($K_2 \to K_3$) is not homotopy-equivalent since the two vertices share an element in their link (encircled in orange).

\subsection{Local Bilinear Jacobi Set Computation}
\begin{figure}[!t]
    \centering
    \includegraphics[width=0.76\columnwidth]{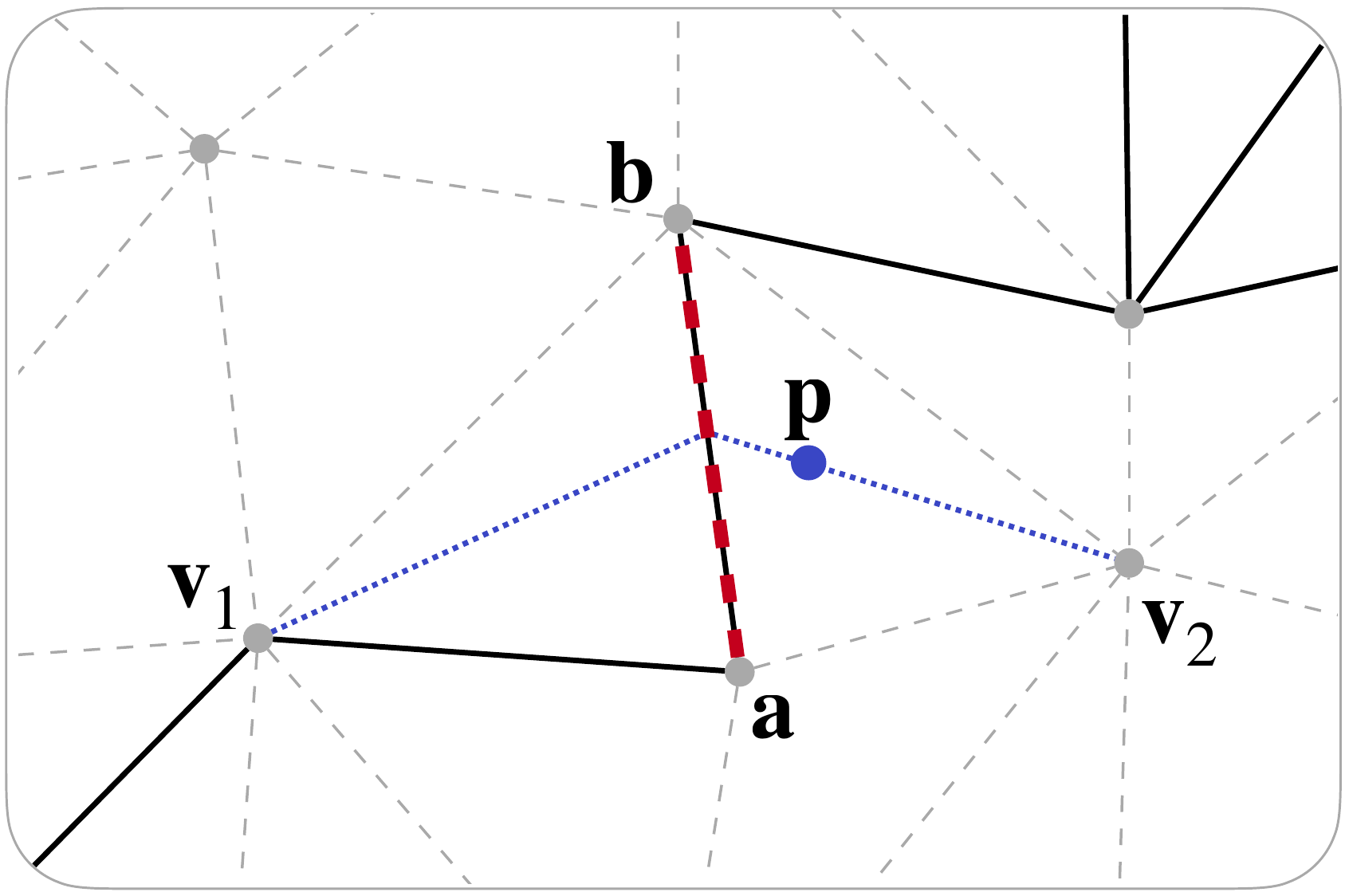}
    \caption{
    Sketch of the PL and LB computation of Jacobi sets.
    The illustration shows the configuration for the edge $\a\b$. While the PL method extracts edges (such as the edge $\a\b$), the LB method computes Jacobi set points (such as the point $p$).
    }
	\label{fig:pl_bl}
\end{figure}

\begin{algorithm}[!b]
    \setstretch{1.1}
	\DontPrintSemicolon
	\KwIn{Scalar fields $f$ and $g$, edges $E=(\vec{e}_i)$}
	\KwOut{$JE = (\vec{e}_i, \vec{p}_i)$: {List of critical edges $\vec{e}_i$ with corresponding Jacobi set points $\vec{p}_i$} }
	\Begin{
		\For{$\vec{e}_i \in E$ \label{alg:JSP_forEdges}}{
			$Find$ $\v,\vv \in \link(\vec{e}_i)$\;
			$Compute$ $\kappa^\lin_{\v}, \kappa^\lin_{\vv}$ \hfill // \cite[Eq.~6 and Eq.~7]{Kloetzl:2022:LocalBilinearJS} \\	\If{$\sig(\kappa^\lin_{\v}) \neq \sig(\kappa^\lin_{\vv})$ }
			{   
			    $Compute$ $\kappa^\bil_{\v}, \kappa^\bil_{\vv}$ \hfill // \cite[Eq.~10 and Eq.~11]{Kloetzl:2022:LocalBilinearJS} \\
				\If{$\sig(\kappa^\bil_{\v}) \neq \sig(\kappa^\bil_{\vv})$ }
			    {$\lambda = \kappa^\bil_{\v} / (\kappa^\bil_{\vv}-\kappa^\bil_{\v})$\;}
			\Else{$\lambda = \kappa^\lin_{\v} / (\kappa^\lin_{\vv}-\kappa^\lin_{\v})$ \;}
			$\m=\a+(\b-\a)/2$\;
			\If{$\lambda < 1/2$}{
			$\vec{p}_i = \v + 2\lambda (\m-\v)$\;
			}
			\Else{$\mathbf{p}_i = \m + (1 - 2\lambda) (\m - \vv)$\;}
				add $(\mathbf{e}_i,\mathbf{p}_i)$ to $JE$
		    }
    			
		}
	}
	\caption{Computation of Jacobi set points \cite[Alg. 1]{Kloetzl:2022:LocalBilinearJS} \label{alg:JSP}}
\end{algorithm}

\begin{figure*}[!t]
    \centering
    \includegraphics[width=0.9\textwidth]{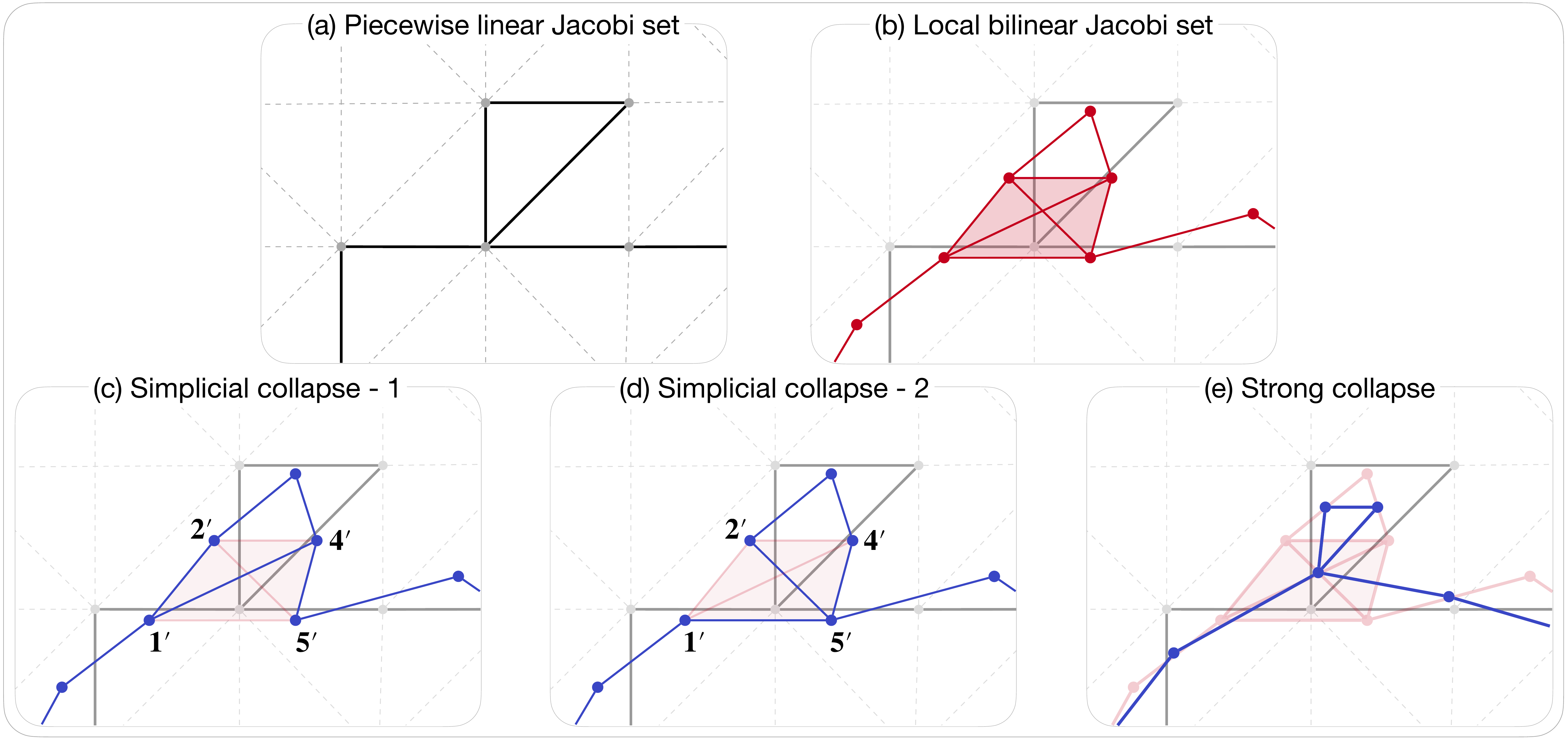}
    \caption{Overview of different Jacobi set representations (top row) and collapsing methods (bottom row). In (a,b), the PL (solid black) and LB (red) methods are shown for a given triangulation. The resulting 1-skeleton of the nerve computed by the LB method is illustrated in (b).
    Simplicial collapse is applied to this Jacobi set in two different ways (c,d). In (e), the strong collapse of the PL method via the geometrical configuration of the LB method is illustrated.}
	\label{fig:LocBil_Ex_Coll}
\end{figure*}

In this subsection, we explain the LB Jacobi set computation by Klötzl et al.~\cite{Kloetzl:2022:LocalBilinearJS}, which is the basis for the connectivity method in this paper.
To do this, we start with a description of the PL computation of Jacobi sets by Edelsbrunner and Harer~\cite{EdelsbrunnerHarer2002}.
For a consistent description of both methods, Klötzl et al.~\cite{Kloetzl:2022:LocalBilinearJS} reformulate the Jacobi set $\Jspace(f,g)=\left\{ \vec{x}\in \Mspace \mid \nabla f(\vec{x}) \times \nabla g(\vec{x}) = 0 \right\}$ for smooth functions $f,g \colon \Mspace \to \Rspace$ defined on a subset $\Mspace \subset \Rspace^2$ into the equivalent set 
\begin{equation}\label{eq:JS_kappa}
\Jspace(f,g)=\left\{ \vec{x}\in \Mspace \mid \kappa_\x(f,g) = 0 \right\},
\end{equation}
where $\kappa_\x(f,g) := \partial_{x} f(\x) \partial_{y} g(\x) - \partial_{y}f(\x) \partial_{x} g(\x)$ is the gradient alignment value.
This value characterizes the linear independence of $\nabla f(\x)$ and $\nabla g(\x)$ at the point $\x\in \Mspace$ and is used to formalize the LB and PL methods as described in the following. 

In general, both methods operate on a triangulation and consider edges individually to extract a 1-manifold that represents the Jacobi set.
In Figure~\ref{fig:pl_bl}, a sketched edge configuration is shown, where the edge $\a\b$ is examined with regard to its criticality.
This is the first part of the LB (or PL) method as shown in Algorithm~\ref{alg:JSP} (lines 2--5).  

Assuming $f$ and $g$ to be PL, i.e., both functions are linear in the triangles $\{\abv\}$ and $\{\avvb\}$, respectively, Klötzl et al.~\cite{Kloetzl:2022:LocalBilinearJS} proved that the (linear) gradient alignment value at the vertices $\v$ and $\vv$ is given by the following formulae
\begin{equation} \label{eq:loc_kappav1}
\kappa_{\v}^\lin(f,g)
=\frac{1}{A_{\abv}} \cdot (f_\b-f_\a)g_\v + (f_\a-f_\v)g_\b + (f_\v-f_\b)g_\a,
\end{equation}
\begin{equation} \label{eq:loc_kappav2}
\kappa_{\vv}^\lin(f,g)
=\frac{1}{A_{\avvb}} \cdot (f_\a-f_\b)g_\vv + (f_\vv-f_\a)g_\b + (f_\b-f_\vv)g_\a,
\end{equation}
where $A_{\vec{p}_1\vec{p}_2\vec{p}_3} = x_{\vec{p}_1}(y_{\vec{p}_2}-y_{\vec{p}_3})+ x_{\vec{p}_2}(y_{\vec{p}_3}- y_{\vec{p}_1}) + x_{\vec{p}_3}(y_{\vec{p}_1} - y_{\vec{p}_2})$ describes the area of the spanned parallelogram between the vertices $\vec{p}_1, \vec{p}_2$, and $\vec{p}_3$.
These gradient alignment values are used to identify the critical edges of the PL approach by Edelsbrunner and Harer:
\begin{equation}\label{eq:equivalence}
\ab \in \Jspace(f,g) \, \Leftrightarrow \, \sig\big(\kappa^\lin_{\v}(f,g)\big) \neq \sig\big(\kappa^\lin_{\vv}(f,g)\big)
\end{equation}
As the PL formulation results in piecewise constant gradient alignment fields $\kappa_\x(f,g)$, it is only possible to visually represent the Jacobi set via the identified critical edges.
For the resulting collection of line segments, the so-called \hyperref[lem:EDL]{Even Degree Lemma} holds~\cite{EdelsbrunnerHarer2002}.
\begin{EDL}
\label{lem:EDL}
	The degree of every vertex in $\Jspace(f,g)$ is even, i.e., the number of critical edges attached to a vertex is even.
\end{EDL}
\noindent Thus, the collection of line segments can be unfolded to a 1-manifold. More precisely, for every vertex with a degree larger than two, the attached edges can be unglued in a way that they do not cross.

Based on the PL approach by Edelsbrunner and Harer~\cite{EdelsbrunnerHarer2002}, the LB method by Klötzl et al.~\cite{Kloetzl:2022:LocalBilinearJS} enhances the representation geometrically while preserving the topology.
This is achieved by introducing the concept of Jacobi set points in the first stage (lines 6--16 in Algorithm~\ref{alg:JSP}) as well as a method to connect them in the second stage. 
Assuming bilinearly interpolated functions $f$ and $g$ on a quadrilateral cell that encloses a critical edge, they showed that the underlying gradient alignment field $\kappa_x$ as a function of $x$ is linear. 
Therefore, it is possible to compute a single point given by the zero of the gradient alignment field between the neighboring vertices of the critical edge.
Klötzl et al.~\cite{Kloetzl:2022:LocalBilinearJS} provided explicit formulae for $\kappa^{\text{bl}}_{\vec{v_1}}$ and $\kappa^{\text{bl}}_{\vec{v_2}}$ to compute the zero for each of the critical cells.
This point is called Jacobi set point and is assigned to the critical edge (in principle, it replaces the critical edge).
We refer to Algorithm~\ref{alg:JSP} for the computation of the Jacobi set points (the algorithm is identical to Algorithm~1 in Klötzl et al.~\cite{Kloetzl:2022:LocalBilinearJS}).

The connection of the Jacobi set points can be done in a canonical way by using the connectivity of the underlying PL method.
An illustration of these steps is provided in Figure~\ref{fig:LocBil_Ex_Coll} (top).
This technique solves the problem of zig-zag patterns but introduces clutter and crossings of line segments.
Klötzl et al.~\cite{Kloetzl:2022:LocalBilinearJS} showed that the homotopy is still equivalent to the PL approach because the connectivity coincides with the topological 1-skeleton of the nerve complex. 
To be more precise, defining $\mathcal{U}$ as the union of the closure of each PL Jacobi set edge, the nerve $\mathcal{N}(\mathcal{U})$ coincides with the local bilinear representation and is, therefore, homotopy-equivalent to the piecewise linear representation~\cite[Appendix]{Kloetzl:2022:LocalBilinearJS}.

\section{Our Approach}
\label{sec:our-approach}

In this section, we present the proposed connectivity method starting with the formulation of desired design goals.
Then, different collapsing strategies are evaluated, which leads to the derivation of our method.
Finally, the associated algorithm is described.

\subsection{Design Goals}
The typical result of the LB method by Klötzl et al.~\cite{Kloetzl:2022:LocalBilinearJS} is demonstrated in Figure~\ref{fig:LocBil_Ex_Coll} (top).
Since the connectivity of the LB method is given by the nerve complex and, in this case, actually consists of a tetrahedron (marked by the red transparent area), the 1-skeleton representation ends up in crossings of line segments.
For this representation, the \hyperref[lem:EDL]{Even Degree Lemma} holds (for any vertex an even number of critical edges is connected to it), resulting in a 1-manifold (the 1-manifold property is achieved through the unfolding of edges).

Our goal is now to find a reduced connectivity, that is, a representation with less clutter due to the conceptual removal of higher-dimensional simplices. 
In addition to that primary goal, the connectivity should satisfy the following design goals as well:
\begin{itemize}
\itemsep0pt
    \item Preservation of topology (homotopy-equivalent representation)
    \item Preservation of the geometrical configuration (Jacobi set points and their positions)
    \item Upholding of the \hyperref[lem:EDL]{Even Degree Lemma}
    \item Computation in a deterministic way.
\end{itemize}

\subsection{Derivation}

We are now able to derive our method.
Starting with the connectivity of the LB method given in Figure~\ref{fig:LocBil_Ex_Coll} (b), we aim for a representation with the desired design goals.
To this end, we explore different collapsing strategies in the following to get rid of the crossings that arise from the higher-dimensional simplices.

\paragraph{Simplicial collapse of LB method}\label{par:simple}

To collapse the tetrahedron in Fig~\ref{fig:LocBil_Ex_Coll} (b) or any other higher-dimensional simplex contained in a simplicial complex, which we will denote as $K$ in the following, we have a closer look at its topological structure.
Due to the \hyperref[lem:EDL]{Even Degree Lemma}, $K$ consists only of odd-dimensional simplices.
Therefore, each of the higher-dimensional simplices can be collapsed individually because, by construction, they are only connected to other simplices via vertices.
This collapsing strategy leads to a 1-complex as illustrated in Figure~\ref{fig:LocBil_Ex_Coll} (c) or, in more detail via elementary simplicial collapses, in Figure~\ref{fig:el_coll}.
In principle, three edges of the topological 1-skeleton are removed for a tetrahedron.

While this approach seems promising with regard to topological and geometrical properties, it is not uniquely determined and may violate the \hyperref[lem:EDL]{Even Degree Lemma}.
This can be observed in Figure~\ref{fig:LocBil_Ex_Coll}~(c) for the vertices $\vec{1}'$ and $\vec{4}'$.
However, this issue can be solved as the following theorem shows.

\begin{theorem}
Given a simplicial complex $K$ containing only odd-dimensional simplices that are connected via vertices. For each $n$-simplex there are $n-1$ configurations to simplicially collapse $K \to L$ such that a $1$-complex is obtained, the \hyperref[lem:EDL]{Even Degree Lemma} holds for $L$, and the vertices of $K$ are preserved. 
\end{theorem}

\begin{proof}
First, we want to point out that each of the vertices in an odd-dimensional $n$-simplex has an odd number of edges inside the $n$-simplex attached to it.
For each higher-dimensional $n$-simplex $S$, ($n>1$), we collapse the simplex iteratively such that all edges of one of the $(n-1)$ contained ($n-1$)-simplices are collapsed (leading to $(n-1)$ possible configurations). 
This ($n-1$)-simplex, which we denote as $S'$, can be chosen freely.
The procedure is illustrated in Figure~\ref{fig:LocBil_Ex_Coll}~(d), where the 3-dimensional simplex $S=\{\vec{1}',\vec{2}', \vec{4}', \vec{5}'\}$ (transparent red tetrahedron) is collapsed via simplicial collapses of the edges contained in the face $S'=\{\vec{1}',\vec{2}', \vec{4}'\}$.

As a result, for each vertex of the ($n-1$)-simplex $S'$ (a vertex of $S'$ has $n$ attached edges in $S$), ($n-1$) edges are collapsed so that only one attached edge remains.
For the vertex that is not contained in $S'$ ($\vec{5}'$ in Figure~\ref{fig:LocBil_Ex_Coll}~(d)), the attached edges do not change.
Therefore, each of the vertices has an odd number of attached edges in $S$ (since 1 and $n$ are odd).
Taking the surrounding simplicial complex into account, where an odd number of edges are attached to each vertex of $S$, we can conclude that every vertex has an even degree.
The 1-complex results by construction through collapsing all edges of one $(n-1)$-simplex $S'$ (resulting in a total of $A(n,1)$, elementary collapses to collapse the simplex $S$ to a $1$-simplex, where $A(n,m):=\sum_{k=0}^{m+1}(-1)^k \binom{n+1}{k} (m+1-k)^n$ is the Eulerian number).
One of the resulting 1-complexes is illustrated in Figure~\ref{fig:LocBil_Ex_Coll}~(d).
\end{proof}

The theorem points out that there are multiple configurations such that the \hyperref[lem:EDL]{Even Degree Lemma} holds.
One of these configurations is shown in Figure~\ref{fig:LocBil_Ex_Coll} (d).
Even though this approach would lead to a topologically and geometrically preserving representation, it is still not uniquely determined, which makes the topological study of scalar fields via Jacobi sets difficult due to ambiguous solutions.

\paragraph{Strong collapse of LB method}

Another collapsing strategy for the tetrahedron in Figure~\ref{fig:LocBil_Ex_Coll} (b) or any other simplicial complex $K$ is the strong collapse.
As pointed out in the background section, an efficient method to perform the strong collapse is to apply the nerve complex twice.
Since we already know that the LB method by Klötzl et al.~\cite{Kloetzl:2022:LocalBilinearJS} coincides with the nerve complex of the PL method by Edelsbrunner and Harer~\cite{EdelsbrunnerHarer2002}, we obtain the desired homotopy-equivalent representation by taking the nerve once again.

For the nerve construction of the LB representation (i.e., the strong collapse of the PL approach), the geometry cannot be preserved as all the simplices are mapped to new vertices.
The heuristic approach is to take the barycenter of the respective simplices as illustrated in Figure~\ref{fig:LocBil_Ex_Coll} (e).

To be clear, the notion of strong collapse refers to the strong collapse of the PL Jacobi set. 
Although this approach preserves the topology (since the PL and LB representations are homotopy-equivalent), satisfies the \hyperref[lem:EDL]{Even Degree Lemma}, and is deterministic, it does not preserve the geometry given by the Jacobi set points.
Ignoring the location of the Jacobi set points leads to a representation that takes no advantage of the LB method.
In fact, the barycentric relocation of all Jacobi set points results in an undesired smoothing that may provide misleading information about the Jacobi set.

\begin{figure}[!t]
    \centering
    \includegraphics[width=\columnwidth]{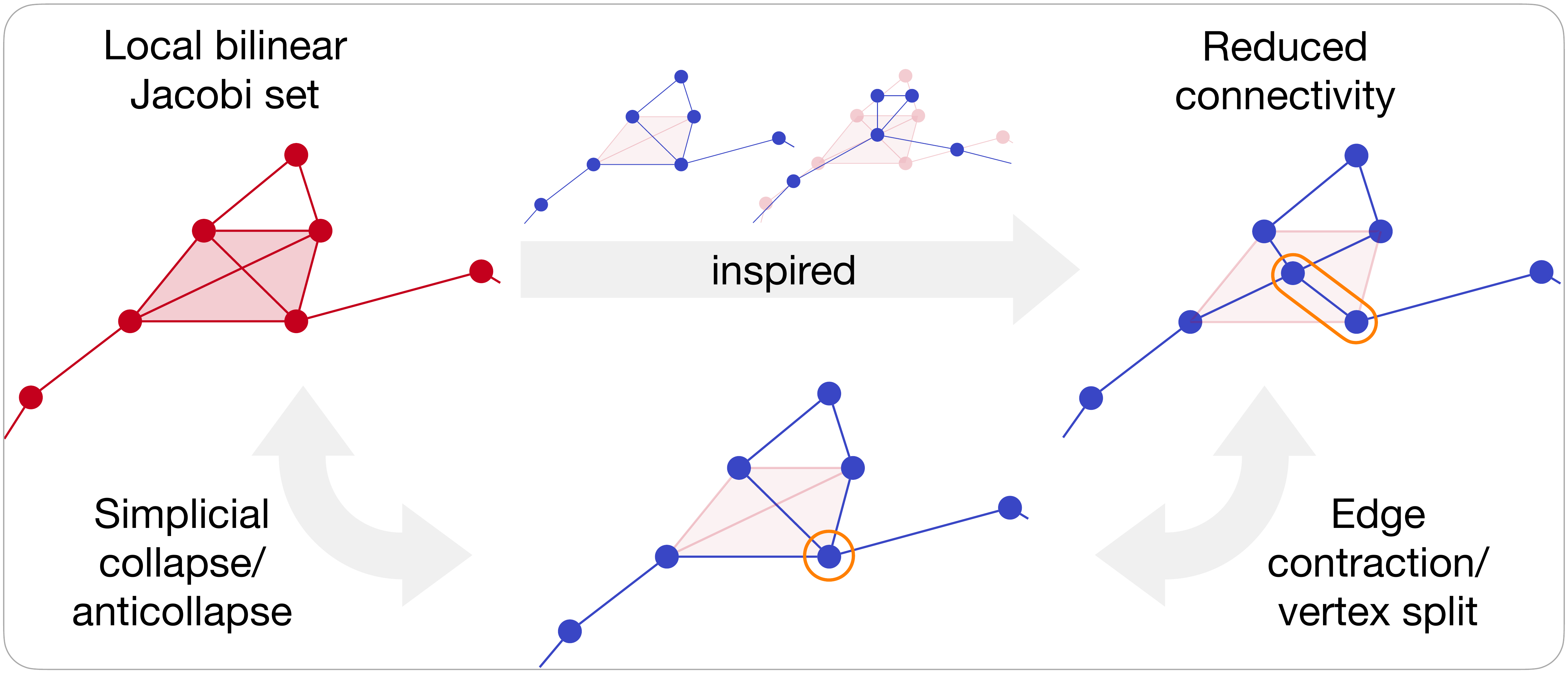}
    \caption{Overview of our method to obtain the reduced connectivity construction. 
    Inspired by the simplicial collapse and strong collapse, the proposed homotopy-equivalent representation is derived. The homotopy-equivalent edge contraction is encircled in orange.}
	\label{fig:collapse_overview}
\end{figure}

\begin{algorithm}[!b]
    \setstretch{1.1}
	\DontPrintSemicolon
	\KwIn{$JE = (\vec{e}_i, \vec{p}_i)$: {List of critical edges $\vec{e}_i$ with corresponding Jacobi set points $\vec{p}_i$}}
	\KwOut{$JS^{\mathcal{C}} = \{l_k = (\vec{p}_{k_i}, \vec{p}_{k_j})\}$: Jacobi set lines}
	\Begin{
	    Initialization list $L$ \hfill // $L$ maps vertices to critical edges \\
	    \For{$(\vec{v}_{i_1},\vec{v}_{i_2})=\vec{e}_i \in JE$ }
	    {
		    $L(\vec{v}_{i_1}) = L(\vec{v}_{i_1}) \cup \{\vec{e}_i\}$\\
		    $L(\vec{v}_{i_2}) = L(\vec{v}_{i_2}) \cup \{\vec{e}_i\}$\\
		}
	    \For{vertex $\vec{v} \in L$}
	    {
	        $L(v) = \{\vec{e}_{v_1},\dots,\vec{e}_{v_d}\}$ \\
    		\If{$d==2$}
    		{
    		add line segment $l=(\vec{p}_{v_1}, \vec{p}_{v_2})$ to $JS^{\mathcal{C}}$ 
    		}
    		\Else
    		{
    		$\vec{p}_{m}^{\mathcal{C}} = \frac{1}{d} \sum_{k=1}^d \vec{p}_{v_k}$\\
    		\For{$j=1,\dots,d$}
    		{
		    add line segment $l=(\vec{p}_{v_j},\vec{p}_{m}^{\mathcal{C}})$ to $JS^{\mathcal{C}}$ \\
		    }
		    }
		}
	}
	\caption{Computation of Reduced Connectivity.\label{alg:hybrid}}
\end{algorithm}

\paragraph{Reduced Connectivity of LB method}

After the evaluation of the two collapsing strategies---simplicial collapse and strong collapse---we can deduce our new connectivity method as a combination of both.
An overview of our approach is given in Figure~\ref{fig:collapse_overview}.
Inspired by the strong collapse, we propose using a barycentric representation for the higher-dimensional simplex.
Nevertheless, inspired by the simplicial collapse, we keep the Jacobi set points and connect the vertices of the higher-simplicial complex to the barycenter.
This leads to a representation with reduced connectivity that inherits the advantages of both collapsing strategies 
while avoiding their issues.

\begin{figure*}[!t]
    \centering
    \includegraphics[width=\textwidth]{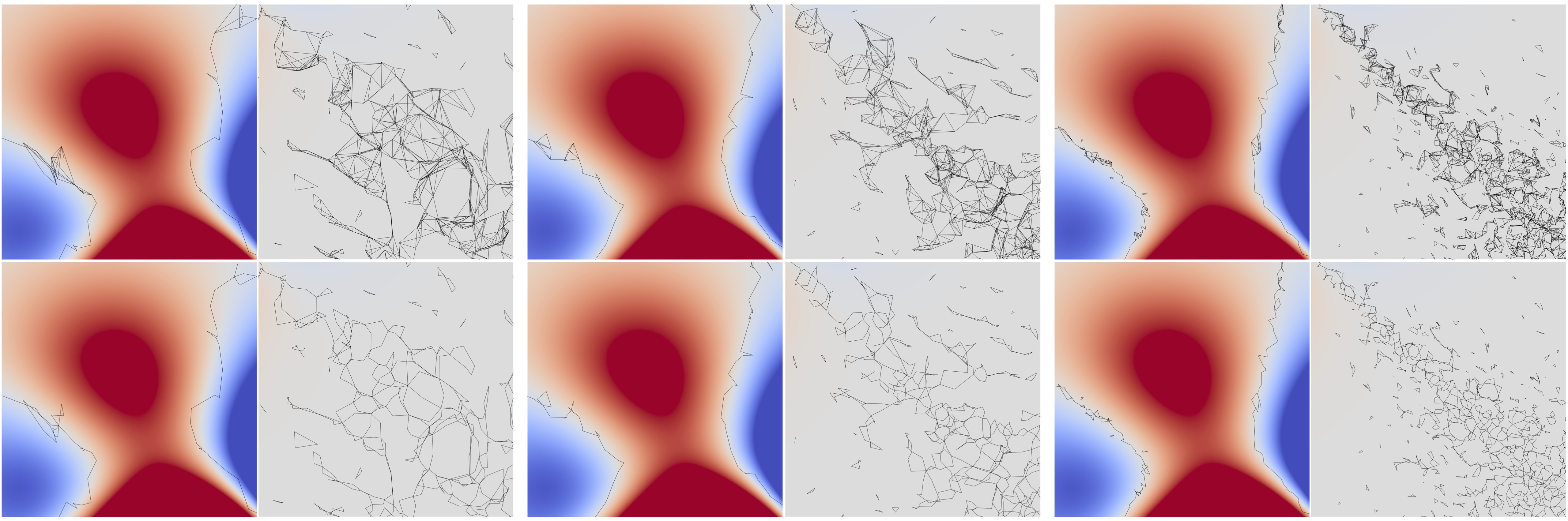}
    \caption{
    Comparison of the original connectivity construction (top row) and our reduced connectivity (bottom row) for the following different resolutions: $60 \times 60$, $80 \times 80$, and $160 \times 160$.
    The dataset, zoomed-in areas, and the color-coding are the same as in Figure~\ref{fig:teaser}.
    }
	\label{fig:analytic}
\end{figure*}

In fact, the representation preserves the geometrical configuration as the Jacobi set points from the local bilinear computation are used.
The barycenter is well justified geometrically as the Jacobi set points around it are zeros of the gradient alignment field (for even higher-dimensional simplices, there are many more Jacobi set points around it).
Also, the representation preserves topology since a homotopy-equivalent simplicial collapse (anticollapse) and edge contraction (vertex split) can be applied to show homotopy equivalence (see Figure~\ref{fig:collapse_overview}).
In addition, the \hyperref[lem:EDL]{Even Degree Lemma} holds because all connections between vertices affected by the higher-dimensional simplex are resolved and, instead, one single connection to the barycenter is established (see Figure~\ref{fig:collapse_overview}).
Finally, the procedure is deterministic by construction. Therefore, all of our design goals are met.
In the next subsection, we formulate an appropriate algorithm to compute the reduced connectivity.

\begin{table}[!b] 
  \centering
\caption{\label{tbl:edge_count}
Number of resulting edges of the non-reduced connectivity method and our reduced connectivity construction.}
    \begin{tabular}{@{}lccc@{}} \toprule
     & \makecell{Non-reduced \\ connectivity} 
     & \makecell{Reduced \\ connectivity} \\
    \midrule
    Analytic (80x80) & 2,014 & 1,424 \\
    Analytic + noise & 6,376 & 4,807\\
    Kármán Vortex & 7,360 & 5,993\\
    Hurricane Isabel & 370,026 & 272,031\\
    Droplet impact & 21,729 & 16,049 \\
    \bottomrule
    \end{tabular}
\end{table}

\subsection{Algorithm}

For the formulation of the algorithm, we recap that Algorithm~\ref{alg:JSP} produces a list of critical edges $\vec{e}_i$ with the corresponding Jacobi set points $\vec{p}_i$.
This list is the input for our algorithm, which is presented in Algorithm~\ref{alg:hybrid}.
According to the derivation of our reduced connectivity, the first step is to identify the higher-dimensional simplices.
These are characterized by vertices that have four or more critical edges connected to them. 
Thus, the first part of Algorithm~\ref{alg:hybrid} is to generate a list $L$ that maps vertices to critical edges (this list is a subset of the general vertex-to-edge list).
Algorithmically, it can be computed via a for loop over the critical edges (lines 2--5).

The next part of the algorithm is to iterate over each vertex that is in the list $L$ and to check how many critical edges are connected to the vertex (lines 6--7).
Due to the \hyperref[lem:EDL]{Even Degree Lemma}, there are either two connected critical edges connected to the vertex or a larger number that is even.
In the case of only two connected critical edges, no higher-dimensional simplex is involved in the representation and, hence, the usual connectivity, i.e., one line segment, can be used (lines 8--9).
If there are more than two critical edges connected to a vertex, then there is a higher-dimensional simplex that needs to be reduced.
According to our derivation, the barycenter of the higher-dimensional simplex is computed via the location of the vertices (line 11) and used for the connectivity, i.e., each vertex is connected to the barycenter via a line segment (lines 12--13). 

Algorithm~\ref{alg:hybrid} results in a collection of line segments that approximates the Jacobi set.
Compared to the representation by Klötzl et al.~\cite{Kloetzl:2022:LocalBilinearJS}, our reduced connectivity produces fewer edges and a clearer representation while still preserving the topology and geometry.
The number of reduced edges can be mathematically formalized as
\begin{equation*}
\sum_{j=1}^K  \binom{N_j}{2}-N_j,
\end{equation*}
where $N_j$ is the dimension of each of the $K$ higher-dimensional simplices.
If there is no higher-dimensional simplex in the representation, then no edges are removed (in this case, there is also no visual clutter).
Otherwise, the number of reduced edges for each higher-dimensional simplex $j$ is given as follows: the binomial $\binom{N_j}{2}$ refers to the number of edges of the $(N_j-1)$-dimensional simplex in the original connectivity method and is subtracted by $N_j$, the number of edges which result from Algorithm~\ref{alg:JSP_forEdges} by connecting each of the $N_j$ vertices to the barycenter (see Figure~\ref{fig:collapse_overview}). 

Another important aspect of Algorithm~\ref{alg:JSP_forEdges} is the low computational overhead compared to the original LB method.
Our connectivity method iterates over all critical edges $\vec{e}_i \in JE$ in the first for-loop and over vertices that are connected to critical edges in the second for-loop.
Since the quantities inside of the loops are computed locally (in our experiments, due to a reasonable triangulation, the value $d$ was typically smaller than 10), the time complexity is linear with respect to the number of critical edges $\#\vec{e}$, i.e., it is given by $\mathcal{O}(\# \vec{e})$.
This is also true for the non-reduced connectivity method by Klötzl et al.~\cite{Kloetzl:2022:LocalBilinearJS}.
Moreover, since the non-reduced connectivity method simply adds line segments between all critical edges that share a vertex (i.e., there is a pairwise connection between the edges $e_{v_1},\dots,e_{v_d}$ in line 7 in Algorithm~\ref{alg:JSP_forEdges}), the only overhead of our connectivity method is introduced through the computation of the barycenter $\vec{p}_m^{\mathcal{C}}$ in line 11.

\begin{table}[!b] 
  \centering
\caption{\label{tbl:compute_times} Computation times (in ms) of the Jacobi set point computation (Alg.~1~\cite{Kloetzl:2022:LocalBilinearJS}), the non-reduced connectivity method (Alg.~2~\cite{Kloetzl:2022:LocalBilinearJS}), and the reduced connectivity construction (Alg.~\ref{alg:JSP_forEdges}).}
    \begin{tabular}{@{}lccc@{}} \toprule
    & \makecell{Comp. of \\ JS points}
    & \makecell{Non-reduced \\ connectivity}
    & \makecell{Reduced \\ connectivity}
    \\
    \midrule
Analytic (80x80) & 10.69 & 0.59 & 1.33\\
Analytic + noise &  17.56  & 2.99 & 8.27 \\
Kármán Vortex & 69.92	 & 2.34 & 2.51 \\
Hurricane Isabel & 473.01  & 75.08 & 162.95 \\
Droplet impact & 1,829.21  &16.34& 42.99\\
    \bottomrule
    \end{tabular}
\end{table}

\section{Evaluation}

\begin{figure*}[!t]
    \centering
    \includegraphics[width=\textwidth]{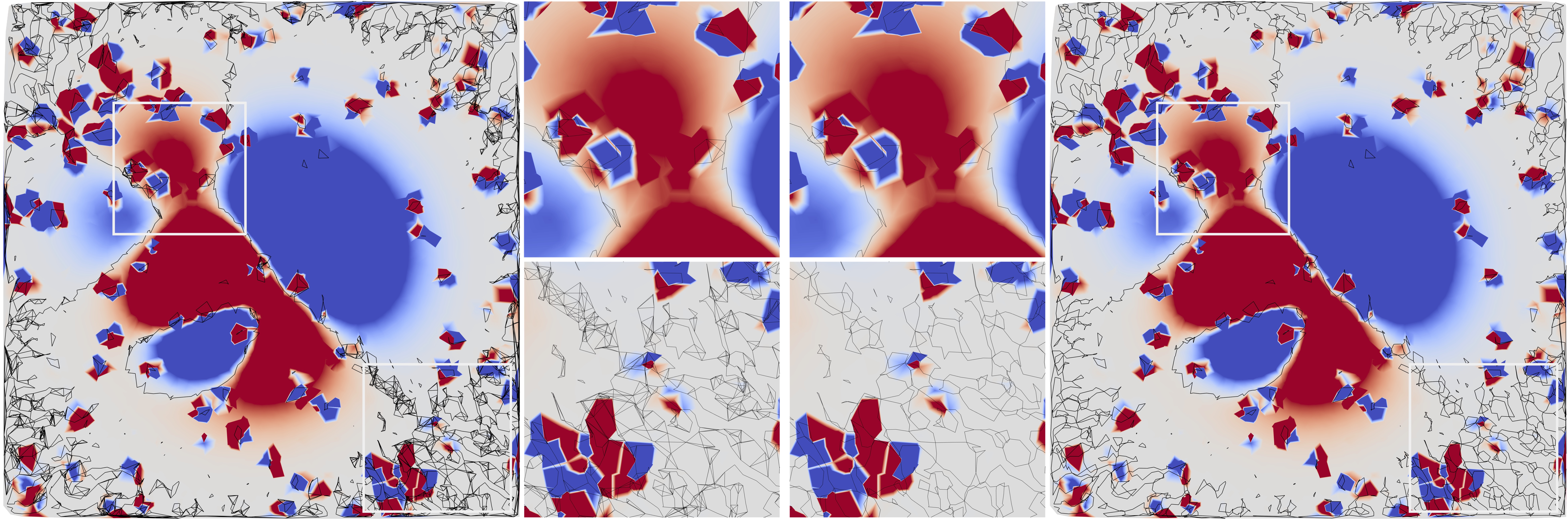}
    \caption{
    Comparison of the original connectivity construction (left) and our reduced connectivity (right) for the analytic dataset with induced uncertainty and the resolution $80 \times 80$.
    The uncertainty is modeled with a salt-and-paper noise and a gaussian noise leading to a distortion of the Jacobi set.
    In the middle, zoomed-in areas are shown that belong to the white-marked areas in the respective outer visualizations.
    }
	\label{fig:analytic_noise}
\end{figure*}

In this section, our new reduced connectivity for the LB computation of Jacobi sets is compared to the visualization by Klötzl et al.~\cite{Kloetzl:2022:LocalBilinearJS}.
We analyze both variants with respect to the following datasets: an analytic example with different resolutions and induced uncertainty as well as a numerically simulated droplet impact on a fluid film. 
In contrast to the analytical dataset, the droplet dataset consists of multiple physical phenomena that highlight the advantages of our method.
For each scenario, the quantitative reduction of edges, as well as the computation time of our connectivity method (Algorithm~\ref{alg:hybrid}), is shown in Table~\ref{tbl:edge_count} and Table~\ref{tbl:compute_times}, respectively.
As a reference, the quantitative analysis is also performed for the von Kármán Vortex Street~\cite{Guenther17} and the Hurricane Isabel\footnote{Hurricane Isabel data produced by the Weather Research and Forecast (WRF) model, courtesy of NCAR and the U.S. National Science Foundation (NSF) (\url{http://www.vets.ucar.edu/vg/isabeldata/})} dataset, which is further investigated in \cite{Kloetzl:2022:LocalBilinearJS}.
The computation was done with MATLAB (R2022a) on a MacBook Pro with an Intel Dual-Core i5 CPU @3.1\,GHz and 8\,GB of RAM.

\subsection{Analytic Dataset}

The first dataset is artificial and consists of two scalar fields $f$ and $g$ that are characterized by bivariate normal distributions on the unit square. 
For a more detailed description and visualization of the scalar fields (as well as the resulting analytic Jacobi set $\mathbb{J}(f,g)$), we refer to Klötzl et al.~\cite{Kloetzl:2022:LocalBilinearJS}, who introduced this dataset.
In the first part of this subsection, we consider the dataset for different resolutions.
Afterward, uncertainty is induced into the dataset to make the extraction and visualization of Jacobi sets more difficult.

\paragraph{Different resolutions}

The dataset is shown for the resolution $40 \times 40$ in Figure~\ref{fig:teaser} and for the resolutions $60 \times 60$, $80 \times 80$, and $160 \times 160$ in Figure~\ref{fig:analytic}.
In Figure~\ref{fig:teaser}, the entire domain is visualized as well as two zoomed-in areas.
The same zoomed-in areas are shown in Figure~\ref{fig:analytic} with a higher resolution.

In all comparisons, we observe that our reduced connectivity successfully reduces the number of edges.
Thus, a less cluttered representation is achieved that shows the connections in a clearer way.
This fact can be particularly observed in the bottom right zoom areas in Figure~\ref{fig:analytic}, where many topological structures are identified. 

The other zoomed-in area in the middle of the unit square demonstrates that our connectivity method removes many redundant edges, i.e., edges that do not lie in the white area (zero level set of the gradient alignment field). These redundant edges arise from the 1-skeleton of the higher-dimensional simplices, which are collapsed by our method.
In Figure~\ref{fig:teaser}, this fact can be observed in the top right part of the first zoomed-in area, where the simplex collapses to an hourglass structure.
The triangle in the lower left part, though, is clearly a discretization artifact and the method does not simplify this, since the edges do not build a higher-dimensional simplex. 
Instead, it maintains the geometrical Jacobi set points and connectivity of the 1-complexes as intended.

Table~\ref{tbl:edge_count} confirms the observations as the number of edges are reduced significantly for the different resolutions.
Approximately $25\%$ of the edges (line segments) are removed in all scenarios due to our connectivity method.
The computation time for this is presented in Table~\ref{tbl:compute_times}.
It can be observed that Algorithm~\ref{alg:hybrid} does only produce little overhead compared to the LB method.

\paragraph{Induced uncertainty}
To modify the analytic dataset with a resolution of $80 \times 80$, we apply a salt-and-paper noise and a weak Gaussian noise to the two scalar fields $f$ and $g$, leading to a distortion of the Jacobi set $\mathbb{J}(f,g)$.
In Figure~\ref{fig:analytic_noise}, the conventional connectivity and our reduced connectivity are shown for the entire domain and the same zoomed-in areas as in Figure~\ref{fig:teaser} (or Fig~\ref{fig:analytic}).

In contrast to the non-modified dataset, we observe that more topological patterns are identified due to distortion, which makes the conventional representation even more cluttered.
Our reduced connectivity still facilitates the representation in a meaningful way and, in particular, has a clearer topological and geometrical identification.
These aspects can be observed in the top zoomed-in areas, where, e.g., loops are visualized in a lucid way.

Analogously to the non-modified dataset, our reduced connectivity method does not produce much overhead (see Table~\ref{tbl:compute_times}), although many more edges are identified.
In fact, Table~\ref{tbl:edge_count} states that the LB method identifies $6,376$ edges and the reduced connectivity configuration only produces $4,807$ edges.
This is a reduction of approximately $25 \%$.

\subsection{Droplet dataset}

\begin{figure}[!t]
    \centering
    \includegraphics[width=\columnwidth]{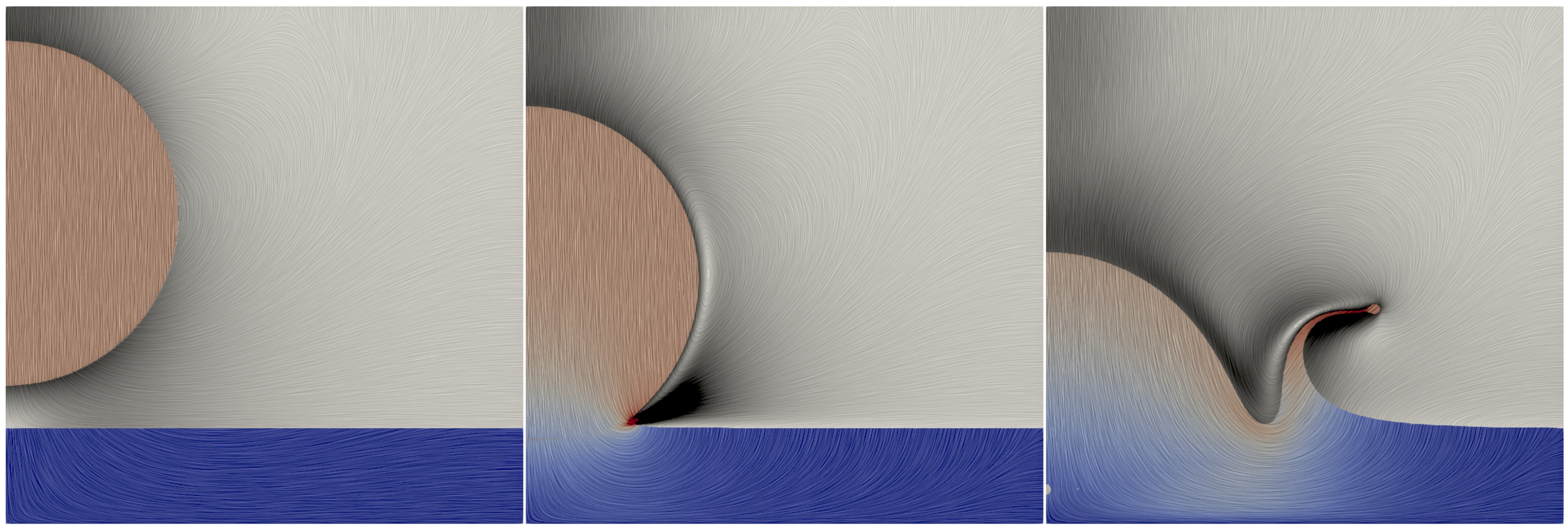}
    \caption{Line integral convolution visualization of the droplet impact dataset at initial condition, shortly after impact, and during crown formation. 
    In general, the color encodes the velocity magnitude of the fluid.
    For gas, the colormap is gray-scaled, whereas, for liquid, a blue-to-red colormap (low to high velocity) is used.}
	\label{fig:droplet_overview}
\end{figure}

\begin{figure*}[!t]
    \centering
    \includegraphics[width=\textwidth]{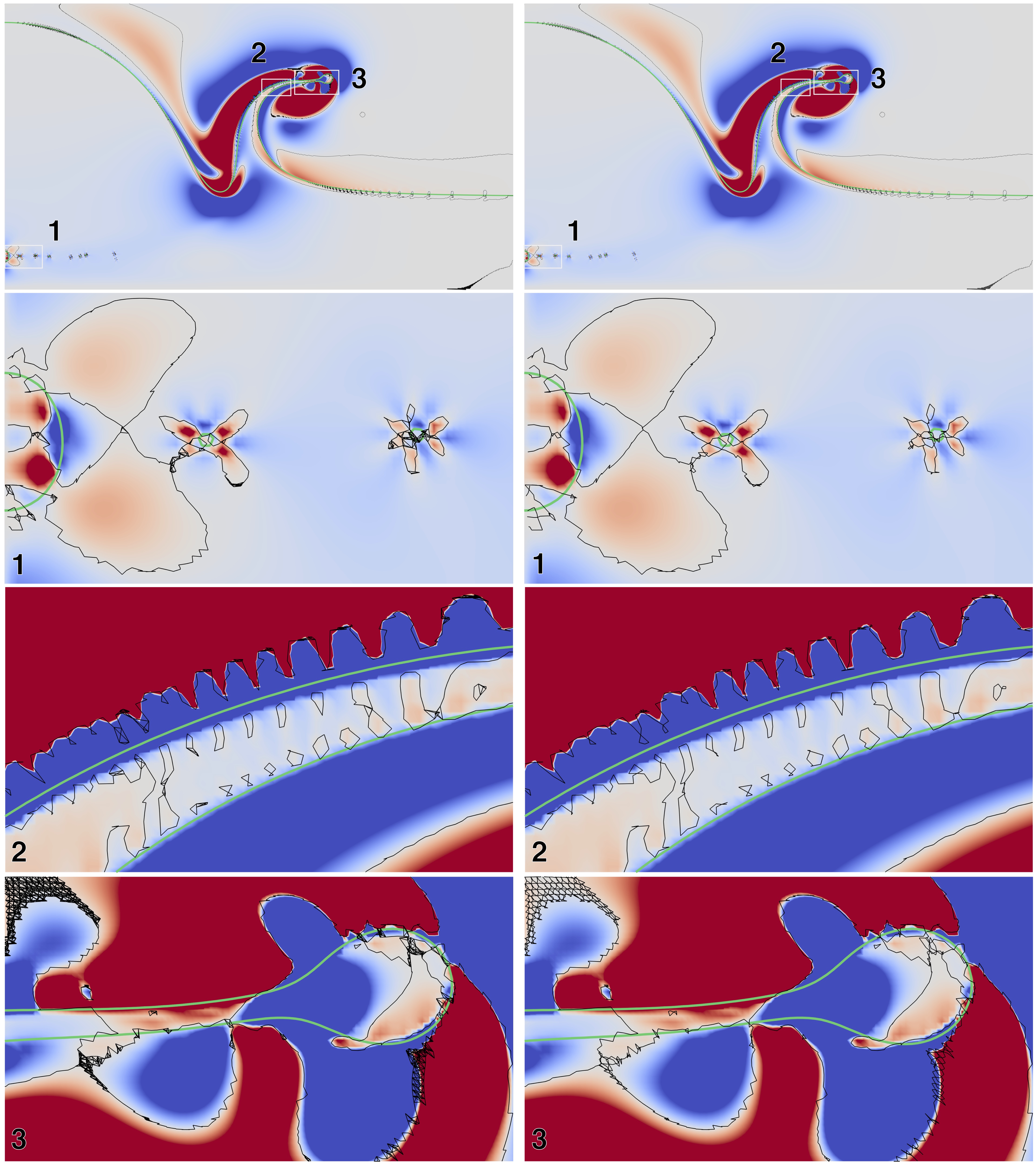}
    \caption{Comparison of the original connectivity construction (left column) and our reduced connectivity (right column) for the LB computed Jacobi set of the droplet dataset.
    The color coding (blue-to-red) shows the gradient alignment field of two consecutive time steps.
    The black solid lines correspond to the extracted Jacobi set and the green solid line indicates the interface between liquid and gas.
    }
	\label{fig:droplet}
\end{figure*}

Droplet impact onto thin wall films is a fundamental process in a lot of modern technological applications and natural processes. 
Over the past years, a lot of experimental and numerical research has been done in this area. 
In this context, the classification of the impact outcome, e.g., splashing with the generation of secondary droplets or deposition, was investigated. 
Furthermore, the crown shape and temporal evolution are of great interest~\cite{roisman_tropea_2002}. 
Besides experimental ~\cite{geppert2017benchmark} and numerical investigations, an analytical model for the evolution of the base radius of the crown was proposed~\cite{lamanna2022drop}.

The utilized Direct Numerical Simulation of a droplet impact onto a thin wall film\footnote{The simulation results of a droplet impact onto a liquid film are publicly available in~\cite{DropletDataDarus}.} was performed with the program package Free Surface 3D~\cite{EISENSCHMIDT2016508}.
The code solves the incompressible Navier-Stokes equations in a one-field formulation with a Volume of Fluid approach. 
It takes all relevant forces into account, being inertia of the liquid and the ambient gas, the surface tension of the interface, as well as friction losses, and gravity.

Exploiting the symmetry of the phenomenon, the 3D simulation has been reduced to a quarter of the droplet. 
The investigated dataset is a slice through the symmetry plane according to Figure~\ref{fig:droplet_overview}.
For each of the 85 time steps, the dataset consists of velocity vectors and a phase indicator function, which distinguishes between liquid and gas.
The resolution is $1024\times 1024$. 
Three snapshots are illustrated in Figure~\ref{fig:droplet_overview} showing the initial condition, a time step shortly after droplet impact, and the formation of the crown.
For the computation of the Jacobi set, the time step 83 (equal to the last snapshot in Figure~\ref{fig:droplet_overview}) is used.
To be more precise, the $x$ and $y$ components (horizontally and vertically) of the velocity field are used as scalar fields $f$ and $g$ as input for the computation of $\mathbb{J}(f,g)$.

Similar to the previous dataset, the number of resulting edges of the reduced connectivity construction is approximately $25\%$ lower, as presented in Table~\ref{tbl:edge_count}. 
Since the droplet dataset is around four times larger than the Hurricane Isabel dataset, the computation of Jacobi set points is more expensive, whereas the non-reduced and reduced connectivity methods require less time (due to the lower number of critical edges), as shown in Table~\ref{tbl:compute_times}.
The qualitative results are shown in Figure~\ref{fig:droplet} (only the lower section of the dataset is used).

Since the visualizations in the top row in Figure~\ref{fig:droplet} do not provide detailed insights into the dataset, three characteristic sections for the droplet impact onto thin wall films are used for the investigation and comparison of the Jacobi set representations.

In the first zoomed-in area a large air bubble and two tiny air bubbles can be observed in the dataset, which are represented by the green line highlighting the interface between liquid and gas. 
Whereas the large bubble is in the center line of impact and is examined in many experimental studies, the tiny bubbles are numerical relicts that are not observed in experiments.
The Jacobi set identifies topological structures around the bubbles.
These structures appear due to the topological change of velocity.
In fact, the larger bubble comes along with much larger topological areas.
As noted before, our method generally unclutters parts of the representation. 
This aspect is clearly visible in the tiny bubble on the right, where the reduced connectivity has a clearer Jacobi representation.

The second zoomed-in area shows a part of the highly curved crown. 
In this region, the liquid has a high velocity. 
The small structures lying inbetween the green interfaces are represented and separated more distinctly with our method. 
In particular, we observe in the top right part a clearer circular topological structure.

Finally, in the third zoomed-in area the thicker rim bounding the thin crown is shown. 
In this area, the surface tension forces of the interface and the inertial forces of the liquid play an important role. 
Here, we observe the same general characteristics of our new reduced connectivity construction: Filled areas get thinned out (top left), 1-manifold geometrical structures get preserved (bottom middle), and small-scale textures are more clearly represented (at the rim bounding the crown). 
In sum, our method enables the reliable and fast identification of topological structures within the droplet impact for further analysis of the phenomena.

\section{Conclusion and Future Work}
\label{sec:conclusion}

This paper introduced a new connectivity method for the local bilinear computation of Jacobi sets~\cite{Kloetzl:2022:LocalBilinearJS}.
The method combines the advantages of different topological collapses to avoid the problem of visualizing higher-dimensional simplices.
The resulting reduced connectivity leads to a representation that is less cluttered while topology and geometry are still preserved.
In addition, important properties such as the \hyperref[lem:EDL]{Even Degree Lemma} still hold and the provided deterministic algorithm comes with only little overhead.
Hence, our proposed reduced connectivity enhances the visualization and analysis of scalar fields via bilinearly computed Jacobi sets.

In the future, we want to use our method for the study of different phenomena including, for instance, fluid dynamics.
In this regard, we want to focus not only on the topological structure but also on the geometrical configuration, which our method indeed accounts for.
Another research direction could be to visualize the higher-dimensional simplices in a completely different way, e.g., to use a color highlighting for the dimensionality such that different connectivity patterns could be emphasized.  

\acknowledgments{The authors thank Anne Geppert and Stefan Schubert for the fruitful discussions about Jacobi sets and droplet impacts during the Droplet Interaction Technologies (DROPIT) summer school 2022.

This research was supported by the Deutsche Forschungsgemeinschaft (DFG, German Research Foundation) through the grants DFG 270852890-GRK 2160/2 and DFG 251654672–TRR 161, the Swedish Research Council (VR) under the grant 2019-05487, the U.S. Department of Energy (DOE) under the grant DOE DE-SC0021015, and the National Science Foundation (NSF) through the grant NSF IIS-1910733.}

\bibliographystyle{abbrv-doi}

\bibliography{main}
\end{document}